\documentclass[11pt,twoside]{article}
\usepackage[utf8]{inputenc}
\usepackage{fancyhead}
\usepackage{psfrag}
\usepackage{epsfig}
\usepackage{graphicx}
\usepackage{amsmath}
\usepackage{amssymb}
\usepackage{amsthm}
\usepackage{amsmath}
\usepackage{xspace}          
\usepackage{subfigure}
\usepackage{cite}
\usepackage{booktabs}
\usepackage{bm}
\usepackage{booktabs}
\usepackage{tabularx}
\usepackage{color}

\usepackage{algorithm}
\usepackage[noend]{algpseudocode}
\usepackage[final]{microtype}

\usepackage{varioref}
\usepackage{cleveref}
\renewcommand{\cref}[1]{\Cref{#1}}

\newtheorem{definition}{Definition}
\newtheorem{proposition}{Proposition}

\newtheorem{lemma}{Lemma}

\newtheorem{problem}{Problem}

\usepackage{IEEEtrantools} 

\newcommand{\Grobner}{Gr\"o{}bner\xspace} 

\newcommand{\word}[1]{\textnormal{#1}}
\newcommand\var[1]{\ensuremath{\mathsf{#1}}}

\newcommand\modop{\ \word{mod}\ }         
\renewcommand\vec[1]{\bm{#1}}            
\newcommand{\floor}[1]{\lfloor #1 \rfloor}

\newcommand{\wdeg}{\deg_w}
\newcommand{\xdeg}{x\!\deg}
\newcommand{\ydeg}{y\!\deg}
\newcommand{\hasse}[2]{\partial^{[#1]} #2} 

\newcommand\F{\mathbb F\xspace} 
\newcommand\ZZ{\mathbb Z\xspace}
\newcommand\NN{\mathbb N\xspace}

\newcommand{\Basis}{\mathcal B}

\newcommand{\Span}[1]{\word{span}(#1)}
\newcommand\Oapp{O^{\hspace*{-.08em}\raisebox{.09em}{$\scriptscriptstyle \sim$}}\hspace*{-.23em}}

\algrenewcommand\alglinenumber[1]{{\scriptsize#1}}   
\algrenewcommand\algorithmicrequire{\textbf{Input:}} 
\algrenewcommand\algorithmicensure{\textbf{Output:}} 
\newcommand{\ass}{\leftarrow}

\newcommand\KNH{K\"otter--Nielsen--H\o{}holdt\xspace}

\begin{document}

\vspace*{5mm}

\noindent
\textbf{\LARGE Fast \KNH Interpolation\\ in the Guruswami--Sudan Algorithm
}
\thispagestyle{fancyplain} \setlength\partopsep {0pt} \flushbottom
\date{}

\vspace*{5mm}
\noindent
\textsc{Johan S.~R.~Nielsen} \hfill \texttt{jsrn@jsrn.dk} \\
{\small Ulm University, Department of Communications Engineering} \\

\medskip

\begin{center}
\parbox{11,8cm}{\footnotesize
\textbf{Abstract.}
  The \KNH algorithm is a popular way to construct the bivariate interpolation polynomial in the Guruswami--Sudan decoding algorithm for Reed--Solomon codes.
  In this paper, we show how one can use Divide \& Conquer techniques to provide an asymptotic speed-up of the algorithm, rendering its complexity quasi-linear in $n$.
  Several of our observations can also provide a practical speed-up to the classical version of the algorithm.
}
\end{center}

\baselineskip=0.9\normalbaselineskip


\section{Introduction}

The computationally most demanding step of the Guruswami--Sudan algorithm \cite{guruswami_improved_1999} is finding a bivariate interpolation polynomial.
Many algorithms have been proposed, both more classical with a quadratic dependence on the code length $n$, e.g.~\cite{nielsen_decoding_1998,lee_list_2008}, as well as approaches utilising fast multiplication methods with a resulting quasi-linear dependence on $n$ \cite{beelen_key_2010,cohn_ideal_2010}.

In this work we show how the \KNH algorithm\footnote{%
  This algorithm is sometimes mistakenly attributed to K\"otter only.
  However, it appeared first in \cite{nielsen_decoding_1998}, stating that it was obtained as a generalisation of an algorithm in K\"otter's thesis.
} of \cite{nielsen_decoding_1998} admits a Divide \& Conquer variant to utilise fast multiplication.
Our algorithm's complexity is $O(\ell^2 s^3 n) + \Oapp(\ell^\omega s n)$, where $\ell, s$ are the \emph{list size} and \emph{multiplicity} parameters.
$\Oapp$ means big-$O$ but with $\log(ns\ell)$ terms omitted, and $\omega$ is the exponent of matrix multiplication, i.e. $\omega \leq 3$.

This is not the fastest possible way to compute an interpolation polynomial, since \cite{cohn_ideal_2010} achieves $\Oapp(\ell^\omega s n)$, but it matches e.g.~the speed of \cite{beelen_key_2010}.
Ours is also a comparatively simple algorithm: for instance, it is trivial to apply the algorithm to K\"otter--Vardy decoding \cite{kotter_algebraic_2003} with varying multiplicities, while this is possible but quite complicated for the lattice-basis reduction approaches of \cite{cohn_ideal_2010,beelen_key_2010,lee_list_2008,alekhnovich_linear_2005}; see \cite{alekhnovich_linear_2005} for a description of how to accomplish this.

The algorithm has been implemented in Sage v.~5.13 and the source code is available at {http://jsrn.dk/code-for-articles.

\section{Preliminaries and the Problem}

First some notation: we will write $\vec 0$ for the all-$0$ matrix, sub-scripted with dimensions.
Likewise $\vec I$ is the identity matrix.
For any matrix $V$, then $V[i,j]$ denotes the $(i,j)$'th entry.
If $V$ is over $\F[x]$ we will write $\deg V$ to denote the greatest degree among the entries of $V$.

For any $Q \in \F[x,y]$ and $w \in \ZZ_+$, denote by $\wdeg Q$ the $(1,w)$-weighted degree of $Q$: $\wdeg x^iy^j = i + wj$ and $\wdeg$ is then extended to polynomials by the maximal of the monomials' $\wdeg$.
$\wdeg$ induces a module monomial ordering $\leq_w$, where ties are broken using the power of $x$.

Let $\F[x,y]_\ell = \{ Q \in \F[x,y] \mid \ydeg Q \leq \ell \}$; this is an $\F[x]$-module, and we will be working with sub-modules of it.
Given a set of polynomials $\Basis \subset \F[x,y]_\ell$ then we denote by $\Span \Basis$ the $\F[x]$-module spanned by $\Basis$.
We will be working with \Grobner bases of such modules, always on the module monomial ordering $\leq_w$.
From now on, this term order is implicit when we say ``\Grobner bases''.

\begin{definition}
  For any $Q \in \F[x,y]$ and point $(x_0,y_0) \in \F^2$, then the $(d_x, d_y)$ Hasse derivative at $(x_0,y_0)$ for $d_x,d_y \in \NN_0$ is
  the coefficient to $x^{d_x} y^{d_y}$ in $Q(x+x_0, y+y_0)$.
  We denote this by $\hasse{d_x,d_y} Q(x_0,y_0)$.

  Let $D_s = \{ (d_x, d_y) \in \NN_0^2 \mid 0 \leq d_x + d_y < s \}$.
  Then we say that $Q$ has a zero of multiplicity at least $s$ at $(x_0, y_0)$ if $\hasse{d_x, d_y} Q(x_0,y_0) = 0$ for all $(d_x,d_y) \in D_s$.
\end{definition}
We will slightly abuse notation in algorithms by using $D_s$ as an ordered list.
The order of $D_s$ is given by $\preceq$: the lexicographical order on integer tuples, i.e.~$(a_1, b_1) \preceq (a_2, b_2)$ if $a_1 < a_2$ or $a_1 = a_2 \land b_1 \leq b_2$.

If $Q = \sum_{i,j} q_{i,j} x^i y^j$ for $q_{i,j} \in \F$, then we have the formula
\[
  \hasse{d_x,d_y} Q(x_0,y_0) = \sum_{i \geq d_x} \sum_{j \geq d_y} \binom i {d_x} \binom j {d_y} q_{i,j} x_0^{i-d_x} y_0^{j-d_y}
\]

Detached from the application in Guruswami--Sudan for decoding Reed--Solomon codes, the interpolation problem that we will solve is the following:
\begin{problem}
  \label{prob:interpolation}
  Given $(x_1, y_1), \ldots, (x_n, y_n) \in \F^2$ with all $x_i$ distinct, as well as $s, \ell, w \in \ZZ_+$, find a $Q \in \F[x,y]$ with $\ydeg Q \leq \ell$ such that $\wdeg Q$ is minimal while $Q$ has a zero with multiplicity at least $s$ at each $(x_i,y_i)$.
\end{problem}

\section{The \KNH algorithm}

The \KNH algorithm (KNH) for solving \cref{prob:interpolation} is very short and given as \cref{alg:knh}.
We will not prove its correctness here but refer the reader to \cite{nielsen_decoding_1998}.
A few comments should be made before we proceed:
Firstly, in \cite{nielsen_decoding_1998}, the initial basis is chosen as $x^iy^j$ for all $0 \leq i < s$ and $0 \leq j \leq \ell$.
It is well-known that one can simply choose $y^j$ for $0 \leq j \leq \ell$, without changing the correctness of the algorithm.

Secondly, the proof of correctness actually establishes that after the $i$'th iteration of the outer loop, $\Basis$ is a \Grobner basis of $M_1 \cap \ldots \cap M_i$, where $M_j \subset \F[x,y]_\ell$ consists of all polynomials with a zero of multiplicity at least $s$ at $(x_j, y_j)$.
Furthermore, each iteration of the inner loop further refines $\Basis$ to be a \Grobner basis of those polynomials $Q$ which also have $\hasse{d_x,d_y} Q(x_i,y_i) = 0$.

\begin{proposition}
  \label{prop:knh_complexity}
  The complexity of the KNH is $O(\ell^2 s^3 n^2)$.
\end{proposition}
\begin{proof}
Firstly, any $b \in \Basis$ has $\xdeg b < sn$ since for each point the same index $t$ can be chosen in \cref{line:knh_choose_b} at most $s$ times.
That is because if for some $d_x$ then $\hasse{d_x-1, d_y} b(x_i, y_i) = 0$ for all $d_y$, then $\hasse{d_x, d_y}((x-x_i)b(x_i, y_i)) = 0$ for all $d_y$; i.e. if some $b_t$ is chosen, it will be replaced with $(x-x_i)b_t$ so it will not be chosen again for the same $d_x$.
Thus the maximal $x$-degree in $\Basis$ increases at most $s$ for every iteration of the outer loop.

Now in the $O(s^2 n)$ iterations of the inner loop, we compute $\ell+1$ Hasse derivatives of basis elements, as well as $\ell$ linear combinations of two basis elements.
By the $\xdeg$ on basis elements, either such operation costs $O(\ell sn)$.
\end{proof}

\begin{algorithm}[t]
  \caption{The \KNH algorithm}
  \label{alg:knh}
  \begin{algorithmic}[1]
    \Require{$\{(x_i,y_i\} \in \F^2, x_i$ distinct. $s,\ell,w \in \ZZ_+$}
    \Ensure{$Q$, a solution to \cref{prob:interpolation}}
    \vspace{.3em}
    \State $\Basis \ass \{ 1, y, \ldots, y^\ell \}$
    \For{$i = 1,\ldots, n$}
        \For{$(d_x,d_y) \in D_s$}
            \State $b_t \ass \arg\min_{b_j \in \Basis} \{ \wdeg b_j \word{ if } \hasse{d_x,d_y} b_j(x_i,y_i) \neq 0 \word{ else } \infty \}$
               \label{line:knh_choose_b}
            \State $\Basis \ass \big\{ b_j - \big(\hasse{d_x,d_y} b_j(x_i,y_i)/\hasse{d_x,d_y} b_t(x_i,y_i)\big) b_t \mid b_j \in \Basis \setminus \{b_t\} \big\}$ \\
               \label{line:knh_update_B}
            \qquad\qquad\qquad $\cup \{ (x-x_i)b_t \}$
        \EndFor
    \EndFor
    \State \Return $\arg\min_{b_j \in \Basis} \{ \wdeg b_j \}$
  \end{algorithmic}
\end{algorithm}

\section{Manipulating Hasse Derivatives}

We will start with some observations on the computation and manipulation of the Hasse derivatives during the inner loop of the algorithm.
For any $Q = \sum_{i=0}^\ell Q_i y^i \in \F[x,y]_\ell$ and $p \in \F[x]$, we will denote by $Q \modop p$ the polynomial $\sum_{i=0}^\ell (Q_i \modop p) y^i$.
Likewise, for a set $\Basis$ of $\F[x,y]$-elements, we will denote by $\Basis \modop p$ the set $\{ b \modop p \mid b \in \Basis \}$.

\begin{lemma}
  \label{lem:dpart_mod}
  For any $Q \in \F[x,y]$, point $(x_0, y_0) \in \F^2$, and $(d_x,d_y) \in D_s$, then
  \[
    \hasse{d_x,d_y} Q(x_0,y_0) = \hasse{d_x,d_y}(Q \modop (x-x_0)^s)(x_0,y_0)
  \]
\end{lemma}
\begin{proof}
  Let $Q = \sum_{i=0}^{\ell} Q_i y^i$ and $\hat Q = \sum_{i=0}^{\ell} \hat Q_i y^i = Q \modop (x-x_0)^s$.
  Then there exist $q_0, \ldots, q_\ell \in \F[x]$ such that $Q_i = \hat Q_i + q_i(x-x_0)^s$ for every $i$.
  But then $Q_i(x+x_0) = \hat Q_i(x+x_0) + x^sq_i(x+x_0)$.
  The lemma now follows from the definition of $\hasse{d_x,d_y}$.
\end{proof}

We need the above lemma for our Fast KNH, but together with another observation it can even be used in the original KNH to speed up calculations:
for each point, we can compute all the Hasse derivatives for each basis element just once and then update them during the iterations of the inner loop.
That is possible since Hasse derivatives change straightforwardly under the operations performed: first, represent the derivatives of a given element $b$ as an upper anti-triangular matrix $H =[ \hasse{d_x, d_y} b(x_i, y_i) ]_{(d_x, d_y) \in D_s} \in \F^{s \times s}$.
Then the linear combinations in \cref{line:knh_update_B} of \cref{alg:knh} can simply be reflected as linear combinations of these Hasse matrices.
Furthermore, if $H_t$ is the Hasse matrix for $b_t$, then the one for $(x-x_i)b_t$ is simply $H_t$ shifted down by one row, and the new first row set to all-zero; the elements now outside the upper anti-diagonal can be set to zero or ignored.

In the original KNH, this can reduce the total cost of computing with Hasse derivatives to $O(\ell s^4 n + sn^2) + \Oapp(\ell s n)$;
due to space limitations we omit the details on this.
Since the cost of updating $\Basis$ in the inner loop still incurs cost $O(\ell^2 s^3 n^2)$, the overall complexity remains unchanged.
However, I can remark that the above optimisation drastically sped up my own software implementation of the KNH algorithm.

This bottleneck of updating $\Basis$ is exactly what is handled in the Fast KNH, described in the next section, allowing us to end up with a complete algorithm which is quasi-linear in $n$.

\section{Fast KNH: A Divide \& Conquer Variant}

The main idea of the Fast KNH is to completely avoid working with the unreduced $\Basis$ in the inner loop, and only work with $\Basis \modop (x-x_i)^s$.
The operations to perform only depend on the Hasse matrices, so we do not actually manipulate $\Basis$ in the inner loop; instead the operations are ``recorded'' as a matrix $T \in \F[x]^{(\ell+1)\times(\ell+1)}$, and  when continuing the interpolation with the next point, they are applied to $\Basis$ as $T(\Basis \modop (x-x_{i+1})^s)$.
In particular, the operations of one iteration of the inner loop can be represented as the matrix $U$:
\begin{IEEEeqnarray}{rCl}
  \label{eqn:update_matrix}
  U &=& \vec I_{(\ell+1)\times(\ell+1)} - \left[ \vec 0_{(\ell+1)\times(t-1)} \mid \vec u_t^\top \mid \vec 0_{(\ell+1)\times(\ell-t+2)} \right] \\
  \vec u_t &=& (H_1[d_x,d_y]/H_t[d_x,d_y], \ldots, H_{t-1}[d_x,d_y]/H_t[d_x,d_y], \notag \\
           &&\qquad 1 - (x-x_i),\  H_{t+1}[d_x,d_y]/H_t[d_x,d_y], \ldots, H_{\ell+1}[d_x,d_y]/H_t[d_x,d_y])
           \notag
\end{IEEEeqnarray}
The list of points $(x_i, y_i)$ to process is then structured into a binary tree to minimise the representation of $\Basis$ necessary at any given time.
This results in two sub-algorithms: \var{InterpolatePoint} as well as \var{InterpolateTree}.
\var{InterpolateTree} is the main entry point, called with the basis $\Basis = \{ 1, y, \ldots, y^\ell \}$.

\begin{proposition}
  \var{InterpolatePoint} is correct.
  It has computational complexity $O(\ell^2 s^3)$.
\end{proposition}
\begin{proof}[Proof sketch]
  Only how the degrees $\delta_j$ are updated has not already been discussed.
  For a given iteration of the loop, let $b_j$ refer to the elements of $T \Basis$ where $T$ is as in the beginning of the iteration, while $b_j'$ are the elements of $T'\Basis$ where $T'$ is $T$ at the end of the iteration.
  Let also $B$ be the set of $b_j$ such that $\hasse{d_x, d_y} b_j(x_i,y_i) \neq 0$, and let $B'$ be the set of $b_j'$ for the same indices.
  Clearly $\wdeg b_t' = \wdeg b_t + 1$, so the update in \cref{line:interpolate_point_delta} is correct.
  We claim $\wdeg b_j' = \wdeg b_j$ for $j \neq t$: by the choice of $b_t$ then $\wdeg b_t \leq \wdeg b_j$, which means that $\wdeg b'_j \leq \wdeg b_j$.
  If $\wdeg b_t < \wdeg b_j$ then clearly $\wdeg b_j' = \wdeg b_j$.
  Otherwise, assume that $\wdeg b_t=\wdeg b_j$.
  Now, $T\Basis$ is a \Grobner basis of $\Span{T\Basis}$ (recall that this was part of the proof of the original KNH in \cite{nielsen_decoding_1998}), so likewise $B$ is a \Grobner basis of $\Span B$.
  Since $\Span{B'} \subset \Span{B}$ then $b_j' \in \Span{B}$.
  But then $\wdeg b_j'$ cannot be less than the $\wdeg$ of all the elements in $B$.

  Due to lack of space, we omit the details on the computational complexity.
\end{proof}

\begin{proposition}
  \var{InterpolateTree} is correct.
  It has computational complexity $O(\ell^2s^3n) + \Oapp(\ell^\omega s n)$, where $n$ is the number of input points.
\end{proposition}
\begin{proof}
  Correctness follows inductively by the correctness of \var{InterpolatePoint}, since
  $\hat \Basis_1 = \hat \Basis \mod \prod_{h=i_1}^t (x-x_i)^s = \Basis \mod \prod_{h=i_1}^t (x-x_i)^s$.

  Let $C(n)$ denote the complexity on $n$ input points, ignoring the costs of calls to \var{InterpolatePoint}.
  Since $\deg T_1, \deg T_2 \leq sn/2$ then $C(n) = 2C(n/2) + \Oapp(\ell^\omega sn/2)$, which means $C(n) \in \Oapp(\ell^\omega sn)$.
  Adding the cost of $n$ calls to \var{InterpolatePoint} yields the result.
  Computing the $O(2n)$ moduli polynomials has negligible cost $\Oapp(sn)$.
\end{proof}

\begin{algorithm}[t]
  \caption{\var{InterpolatePoint}}
  \label{alg:interpolate_point}
  \begin{algorithmic}[1]
    \Require{$(x_i,y_i) \in \F^2$, $s,\ell,w \in \ZZ_+$, $\hat \Basis$, and $\{\delta_j\}_j$.
      Here $\hat\Basis = \Basis \mod (x-x_i)^s$ where $\Basis \subset \F[x,y]_\ell$ is a \Grobner basis of $\Span \Basis$, and $\delta_j = \wdeg b_j$ for each $b_j \in \Basis$.}
    \Ensure{$T \in \F[x]^{(\ell+1) \times (\ell+1)}, \{\hat \delta_j\}_j$.
      Here $T\Basis$ is a \Grobner basis of $\Span\Basis \cap M_i$, and $\hat \delta_j = \wdeg \hat b_j$ for each $\hat b_j \in T\Basis$.}
    \vspace{.3em}
    \State $H_j = [\hasse{d_x,d_y} b_j(x_i,y_i)]_{(d_x,d_y) \in D_s}$ for each $\hat b_j \in \hat\Basis$
    \State $T = \vec I_{(\ell+1)\times(\ell+1)}$
    \For{$(d_x,d_y) \in D_s$}
        \State $t \ass \arg\min_{t \in \{1,\ldots,\ell+1\}} \{ \delta_j \word{ if } H_j[d_x,d_y] \neq 0 \word{ else } \infty \}$
        \State $H_j = H_j - (H_j[d_x,d_y]/H_t[d_x,d_y])H_t$, for $j \neq t$
        \State $H_t = [\ \vec 0_{(\ell+1) \times 1} \mid \grave H_t^\top \ ]^\top$ where $\grave H_t$ is $H_t$ with the last row removed
        \State $T = UT$, where $U$ is as in \eqref{eqn:update_matrix}
        \State $\delta_t = \delta_t + 1$ \label{line:interpolate_point_delta}
    \EndFor
    \State \Return $T, \{ \delta_j \}_j$
  \end{algorithmic}
\end{algorithm}

\begin{algorithm}[!]
  \caption{\var{InterpolateTree}}
  \label{alg:interpolate_tree}
  \begin{algorithmic}[1]
    \Require{$(x_{i_1},y_{i_1}\}, \ldots, (x_{i_2},y_{i_2}) \in \F^2$, $s,\ell,w \in \ZZ_+$, $\hat \Basis$ and $\{\delta_j\}_j$.
      Here $\hat\Basis = \Basis \mod \prod_{h=i_1}^{i_2}(x-x_h)^s$ where $\Basis \subset \F[x,y]_\ell$ is a \Grobner basis of $\Span \Basis$, and $\delta_j = \wdeg b_j$ for each $b_j \in \Basis$.}
    \Ensure{$T \in \F[x]^{(\ell+1) \times (\ell+1)}, \{\hat \delta_j\}_j$.
      Here, $T\Basis$ is a \Grobner basis of $\Span \Basis \cap M_{i_1} \cap \ldots \cap M_{i_2}$ and $\hat \delta_j = \wdeg \hat b_j$ for each $\hat b_j \in T\Basis$.}
    \vspace{.3em}
    \If{$i_1 = i_2$}
      \Return $\var{InterpolatePoint}((x_{i_1},y_{i_1}),\ \hat \Basis,\ \{\delta_j\}_j)$
    \Else
      \State $t \ass \floor{(i_1+i_2)/2}$;\quad $\hat \Basis_1 \ass \hat\Basis \modop \prod_{h=i_1}^t(x-x_h)^s$
      \State $(T_1, \{\delta_j\}) \ass \var{InterpolateTree}((x_{i_1},y_{i_1}),\ldots,(x_t,y_t),\ \hat \Basis_1,\ \{\delta_j\}_j)$
        \label{line:interpolate_tree_first}
      \State $\hat \Basis_2 \ass T_1\hat\Basis \modop \prod_{h=t+1}^{i_2}(x-x_h)^s$
      \State $(T_2, \{\delta_j\}) \ass \var{InterpolateTree}((x_{t+1},y_{t+1}),\ldots,(x_{i_2},y_{i_2}),\ \hat\Basis_2,\ \{\delta_j\}_j)$
      \State \Return $(T_2T_1, \{\delta_j\})$
    \EndIf
  \end{algorithmic}
\end{algorithm}

{
\let\oldbibliography\thebibliography \renewcommand{\thebibliography}[1]{%
  \oldbibliography{#1}%
  \setlength{\parskip}{3pt}%
  \setlength{\itemsep}{0pt}%
}
\bibliographystyle{abbrv}
\bibliography{bibtex}
}

\end{document}